\documentclass[hidelinks,11pt,english]{article} % this hide the green boxes on the citations
\usepackage{geometry}                % See geometry.pdf to learn the layout options. There are lots.
\usepackage{graphicx}
\usepackage{amssymb}
\usepackage{amsmath}
\usepackage{epstopdf}
\usepackage[dvipsnames]{xcolor}
\usepackage{natbib}
\usepackage{algorithm}
\usepackage{hyperref}
\newcommand{\bean}{\begin{eqnarray*}}
\newcommand{\eean}{\end{eqnarray*}}
\newcommand{\bea}{\begin{eqnarray}}
\newcommand{\eea}{\end{eqnarray}}
\newcommand{\beaa}{\begin{array}}
\newcommand{\eeaa}{\end{array}}
\usepackage{apxproof}
\usepackage{xargs}
\usepackage[colorinlistoftodos,prependcaption,textsize=tiny]{todonotes}
\newcommandx{\unsure}[2][1=]{\todo[fancyline,backgroundcolor=red!25,bordercolor=red,#1]{#2}}
\newcommandx{\change}[2][1=]{\todo[linecolor=blue,backgroundcolor=blue!25,bordercolor=blue,#1]{#2}}
\newcommandx{\info}[2][1=]{\todo[linecolor=OliveGreen,backgroundcolor=OliveGreen!25,bordercolor=OliveGreen,#1]{#2}}
\newcommandx{\improvement}[2][1=]{\todo[linecolor=Plum,backgroundcolor=Plum!25,bordercolor=Plum,#1]{#2}}
\newcommandx{\thiswillnotshow}[2][1=]{\todo[disable,#1]{#2}}
\newcommand{\cR}{\mathcal{R}}

\newtheoremrep{thm}{Theorem}  % Environment to duplicate Theorem statements in the main text and also in the Appendix
\newtheoremrep{lemma}{Lemma}  % Environment to duplicate Lemma statements in the main text and also in the Appendix
\newtheoremrep{conjecture}{Conjecture}
\newtheoremrep{prop}{Proposition}
\newtheoremrep{cor}{Corollary}
\newtheoremrep{example}{Example}
\begin{document}

\title{Bounds and Heuristics for Multi-Product Personalized Pricing}

\author{Guillermo Gallego\footnote{Industrial Engineering and Decision Analytics Hong Kong University of Science and Technology, Kowloon, Hong Kong, ggallego@ust.hk. Supported by RGC project 16211619. }
\and
Gerardo Berbeglia\footnote{Melbourne Business School, The University of Melbourne, Australia, g.berbeglia@mbs.edu.}
}
%\date{}                                           % Activate to display a given date or no date
\maketitle

\begin{abstract}
We present tight bounds and heuristics for personalized, multi-product pricing problems. Under mild conditions we show that the best price in the direction of a positive vector results in profits that are guaranteed to be at least as large as a fraction of the profits from optimal personalized pricing.  For unconstrained problems, the fraction depends on the factor and on optimal price vectors for the different customer types. For constrained problems the factor depends on the factor and a ratio of the constraints. Using a factor vector with equal components results in uniform pricing and has exceedingly mild sufficient conditions for the bound to hold. A robust factor is presented that achieves the best possible performance guarantee.  As an application, our model yields a tight lower-bound on the performance of linear pricing relative to optimal personalized non-linear pricing, and suggests effective non-linear price heuristics relative to personalized solutions. Additionally,  our model provides guarantees for simple strategies such as bundle-size pricing and component-pricing with respect to optimal personalized mixed bundle pricing. Heuristics to cluster customer types are also developed with the goal of improving performance by allowing each cluster to price along its own factor. Numerical results are presented for a variety of demand models that illustrate the tradeoffs between using the economic factor and the robust factor for each cluster, as well as the tradeoffs between using a clustering heuristic with a worst case performance of two and a machine learning clustering algorithm. In our experiments economically motivated factors coupled with machine learning clustering heuristics performed best.

\end{abstract}

\section{Literature Review and Summary of Contributions}

Studies of pricing and demand estimation go back to 17th century \citet{davenant1699} with significant work done in the 19th century by \citet{cournot1838} and others economists of that era.  Optimal pricing continues to attract researchers and practitioners  to develop a deeper understanding and make price optimization more practical. Personalized pricing, a form of third degree price discrimination  has emerge as more data is available, computing power becomes cheaper, and more transactions are done electronically enabling firms to charge different prices to different customer types. One key question is how much better is personalized pricing over non-personalized pricing. Another is to develop practical heuristics that perform well in practice and have a tight worst-case performance guarantee relative to more sophisticated pricing policies.

For the single product case, \citet{bergemann2020uniform} present a worst case bound  of two for the ratio of personalized to uniform pricing,  when the profit function for each customer type  is concave and the demand functions take positive values over a common compact set. \citet{malueg2006bounding} under more mild assumptions obtain a bound equal to number of types. \citet{elmachtoub2020value} obtains tight and robust bounds that depends on summary statistics of the aggregate demand distribution. \citet{gallegotopaloglu} provide bounds and heuristics for specific families of demand distributions including linear, exponential and logistic demand functions. They also include robust procedures to cluster customers types. \citet{chen2019distribution} present results for single-product distribution-free pricing.

In this paper we provide bounds and heuristics for the multi-product pricing problem including bounds for optimal pricing under a single factor (with uniform pricing as a special case) and show how our results can be used to bound the performance of linear-pricing versus non-linear pricing for the single product case. Our work can be seen as a  generalization of results in \citet{berbeglia2020assortment} that provide performance guarantee of uniform pricing for a subset of multi-product demand pricing problems known as \emph{envy-free pricing}. Besides extending the performance guarantees of uniform pricing to a much broader class of single factor models, and to linear versus non-linear pricing, our performance guarantees are also stronger because they hold with respect to the optimal personalized pricing profit rather than the optimal non-personalized profit.

%Personalized pricing: \citet{chen2017introduce}; \citet{chen2020competitive} ,\citet{elmachtoub2018value} , ,

%Personalized pricing:

\section{Single Factor versus Personalized Pricing}

Consider a firm with $n$ products and $m$ customers types. Let $d_{ij}(p)$
denote the demand for product $i \in  N: = \{1,\ldots, n\}$ for customer type $j \in M : = \{1,\ldots,m\}$
at price vector $p: = (p_1,\ldots,p_n)$. We assume that all demand functions are non-negative.  The profit function\footnote{After a simple transformation $p \leftarrow p-c$ and $d(p) \leftarrow d(c+p)$ if there is a non-zero unit cost vector $c$.} for type $j$ customers is given by
$R_j(p) :=\sum_{i \in N}p_id_{ij}(p)$. Denote by $\cR^*_j := \max_{p \geq 0}R_j(p)$ the maximum profit for type $j$ customers and let  $\theta_j > 0, j \in M$ be distribution of the customer types, so $\sum_{j \in M} \theta_j = 1$. Then $\bar{\cR} := \sum_{j \in M} \theta_j \cR^*_j$ is the optimal profit from personalized pricing, also known as third-degree price discrimination.  Let  $R(p) := \sum_{j \in M} \theta_j R_j(p)$ be the profit over all types at price $p$ and let $\cR^* = \max_{p \geq 0}R(p)$. We will refer to $\cR^*$ as the optimal profit from non-personalized pricing. Solving for $\cR^*$ may be difficult even when solving for $\cR^*_j, j \in M$ is easy. This is  because the aggregate demand function may be significantly more complex than the underlying demands for the customer types.  As a result, heuristics are often used. Here we consider a class of single factor heuristics (or pricing policies) where pricing is done along a positive vector $f$, with $\cR^f := \max_{q > 0} R(q f)$. This is a single dimensional optimization problem that can be solved numerically. Notice that $f = e$, the vector of ones,  results in the uniform pricing policy with $\cR^e$ the optimal profit under uniform pricing. Clearly $\cR^f \leq \cR^*$ with equality holding if $f \in \arg\max_{p \geq 0} R(p)$. Several questions arise from this setting, including finding tight bounds of the form (1) $\cR^* \leq \beta \cR^f$ that  provide performance guarantees for simple heuristics (including uniform pricing) for non-personalized pricing, or (2) $\bar{\cR} \leq \beta \cR^*$ to assess the benefit of personalized over non-personalized pricing. While both of these questions have been partially answered for specific pricing models, no tight bounds of type (1) and (2) are known for general pricing models with more than one product. To answer these questions simultaneously we will provide a tight upper-bound of the form $\bar{\cR} \leq \beta \cR^f$  which implies the former bounds on account of $\cR^* \leq \bar{\cR} \leq \beta \cR^f \leq \beta \cR^*$.

{\bf Brief  preview our results}:  Let $\bar{p}_{ij}, i \in N$ be an optimal price vector for type $j \in M$. If there are positive scalers $h$ and a $k$ such that $hf_i \leq  \bar{p}_{ij} \leq h\exp(k)f_i$ for all $i \in N, j \in M$, then under mild conditions $\bar{\cR} \leq (1+ k) \cR^f$. For the special case $f = e$,  if $\bar{p}_{ij} \in [h, 7.38h]$ for some $h > 0$, then $\bar{\cR} \leq 3\cR^e$. A sufficient condition for $f = e$ is that the products are week substitute and have the connected substitute property for each $j \in M$, see \citet{BGH}. The bounds work verbatim if prices are constrained to the stated intervals.\\

%We remark that in the case of a single market segment with $\cR^*$ difficult to compute, the bound $\bar{\cR} \leq \beta\cR^f$ reduces to $\cR^* \leq \beta \cR^f$. \\

%The latent class MNL is an example of this, see \citet{hanson1996optimizing}, and even the linear demand model presents some difficulties as we will discuss in  \S\ref{sec:app}.

 \noindent{\bf Assumption~0 (A0)}:  We assume that optimal prices are positive and finite. \\
 Assumption~0 (A0) holds for virtually all practical pricing problems as few firms price their goods at zero when optimizing profits and there is no demand at infinite prices.  In our analysis we will first obtain a bound for the unconstrained case assuming that we are able to solve the pricing problem for each market segment.  We  then consider the case where prices are constrained to compact sets of the form $p_i \in [f_iq_{\min} , f_i q_{\max}], i \in N, j \in M$ for {\em exogenously} given $0 < q_{\min} < q_{\max} < \infty$.

 Let $\bar{p}_{ij}, i \in N$ is a vector of optimal prices for $R_j(p), j \in M$. Define $\delta_{ij}(q) := 1$ if $q \leq \bar{p}_{ij}/f_i$ and $\delta_{ij}(q) := 0$ otherwise. Let $\bar{p}^j$ be the vector with components $\bar{p}_{ij},  i \in N$ and consider
 $$G(q): = \sum_{j \in M} \theta_j \sum_{i \in N} f_id_{ij}(\bar{p}^j)\delta_{ij}(q)~~~q \geq 0.$$
 Then $G(q)$ is the $f$-weighted demand at the personalized optimal solution $\bar{p}^j, j \in M$ filtering out $i,j$ combinations for which $\bar{p}_{ij} < qf_i$, so some combinations are dropped as $q$ increases. On the other hand consider
 $$H(q): = \sum_{j \in M} \theta_j \sum_{i \in N} f_id_{ij}(qf)$$
 be the $f$-weighted demand at price vector $qf$. We are now ready to state our main assumption:\\

 \noindent{\bf Assumption~1 (A1)}: $G(q) \leq H(q)$ for all $q \geq 0$. \\

 For any given $q$, the LHS of A1 accumulates the $f$-weighted demands at optimal personalized pricing for $i,j$ combinations with $\bar{p}_{ij} \geq qf_i$ whereas the RHS is the $f$-weighted demand over all products and market segments at price vector $qf$. In many cases the inequality holds even if we filter terms on the right hand side by $\delta_{ij}(q)$. Later we will provide  sufficient conditions for A1, but we state here the weaker A1 to highlight the generality of our results.

 To gain intuition of how this inequality will be used, define $G_j(q) := \sum_{i \in N} f_id_{ij}(\bar{p}^j)\delta_{ij}(q)$ for $q \geq 0$ and notice that $\cR^*_j = \int_0^{\infty}G_j(q)dq$. Multiplying by $\theta_j$ and adding over $j \in M$ we conclude that $\bar{\cR} = \int_0^{\infty}G(q)dq$.
 Notice that $\delta_{ij}(q) =1$ for all $q \in [0,q_{\min}]$ and $\delta_{ij}(q) = 0$ for all $q  > q_{\max}$,
 where  $q_{\min} := \min_{i \in N, j \in M}\bar{p}_{ij}/f_i$ and $q_{\max}: = \max_{i \in N, j \in M}\bar{p}_{ij}/f_i$. This implies that
 $$\bar{\cR} = q_{\min}G(q_{\min}) + \int_{q_{\min}}^{q_{\max}}G(q)dq.$$
 Finally, notice that $R(qf) = qf'd(qf) = qH(q)$.  We are now ready to state our main result.
 \begin{thm}
\label{thm:sf}
Suppose that A0 and A1 hold.  Then $\bar{\cR} \leq \beta \cR^f$ where $\beta := 1 + \ln(q_{\max}/q_{\min})$.
Moreover, the bound is tight.
\end{thm}

\begin{proof}
%By construction $\bar{p}_{ij}/f_i \in [q_{\min}, q_{\max}]$. This implies that $\int_0^{q_{\min}}\delta_{ij}(q)dq = q_{\min}$ and $\int_0^{q_{\max}}\delta_{ij}(q)dq = \bar{p}_{ij}/f_i$. Consequently,
% \bean
%\cR^*_j & = & \sum_{i \in N} p_{ij}d_{ij}(\bar{p}^j) =  \int_0^{q_{\max}} G_j(q)dq \\
%& = &  q_{\min}G_j(q_{\min})  + \int_{q_{\min}}^{q_{\max}} qG_j(q) \frac{dq}{q}
%\eean
%holds for all $j \in M$. Then
 \bean
\bar{\cR} & = & q_{\min}G(q_{\min}) +  \int_{q_{\min}}^{q_{\max}} q G(q) \frac{dq}{q}\\
& \leq &  R(q_{\min} f) + \int_{q_{\min}}^{q_{\max}} R(qf) \frac{dq}{q}\\
& \leq & \cR^f\left[1 + \int_{q_{\min}}^{q_{\max}} \frac{dq}{q} \right]\\
& = & \cR^f[1  + \ln(q_{\max}/q_{\min})] = \beta \cR^f
\eean
where the first equality follows form $G(q) = qG(q)/q$ for $q > q_{\min}> 0$; the first inequality follows from $qG(q) \leq qH(q) = R(qf)$ due to A1. The second inequality follows from $R(qf) \leq \cR^f$.

We now address the tightness of the bound. Suppose there are a continuum of types with willingness to pay in the interval $[1, \rho]$ for some $\rho > 1$. Assume that tail of the distribution of types is given by $d(p) = k\min(1, 1/p)$ over $p \in [0, \rho]$ with $k = 1/\beta$ and $\beta = 1 + \ln(\rho)$, so that it integrates to one. Then, personalized pricing results in expected profit $1$, so $\bar{\cR} = 1$. For uniform pricing, any price in the interval $[1, \rho]$ is optimal, resulting in expected profit $\cR^e = 1/\beta$. Clearly $\bar{\cR}/\cR^e = \beta = 1 + \ln(\rho)$ so the bound is tight.

\end{proof}

\begin{cor}
Suppose the firm imposes the constraint $p_{ij} \in [q_{\min} f_i, q_{\max}f_i]~~\forall i \in N, j \in M$
for exogenously given  $f$, $q_{\min} > 0$ and  $q_{\max} = \exp(k)q_{\min}$ for some $k > 0$. If A0-A1 hold then  $\bar{\cR} \leq (1 + k)\cR^f$ where $\bar{\cR}$ is the maximum personalized profit subject to the stated constraints and $\cR^f$ is the optimal price along $qf$ constrained to $q \in [q_{\min}, q_{\max}]$.
\end{cor}
\begin{proof}
The result follows directly from Theorem~\ref{thm:sf} since $\ln(q_{\max}/q_{\min}) = k$ resulting in $\beta = 1 +k$. Consequently $\bar{\cR} \leq (1+k)\cR^f$.
\end{proof}

We remark that Theorem~1 also holds for continuous customer types as long as $q_{\min}$ is bounded away from zero and $q_{\max}$ is finite.  Slightly sharper bounds can be obtained if the set of allowable prices is a finite set bounded away from zero. As an example, if $Q = \{q_1, \ldots, q_K\}$ with $q_k$ strictly decreasing in $k$, we obtain $\beta = \sum_{k=1}^K(q_k-q_{k+1})/q_k \leq K$ where for convenience we set $q_{K+1}: = 0$.  It is also possible to prove that the sharpened bound is tight regardless of the number of customer types. The proof follows the same logic replacing sums with integrals and changing the order of summation and is omitted for brevity.

% and a closed-form formula for $f$ as a function of $\bar{p}^j, j \in M$,  that minimizes $\beta$.

\begin{prop}
\label{prop:f-reg}
The following two properties are sufficient conditions for A$1$.
\begin{itemize}
\item[P1] $d_{ij}(p)$ is increasing\footnote{We use the terms increasing and decreasing in the weak sense unless otherwise stated.} in $p_k$ for all $k \neq i$, for all $j \in M$.
\item[P2] $\sum_{i \in N} f_i d_{ij}(p)$ is decreasing in $p_i$ for all $i \in N$ and all $j \in M$.
\end{itemize}
\end{prop}

We remark that A1 need only hold at $\bar{p}^j, j \in $ so conditions P1 and P2 are much stronger than needed. If $f = e$, then P1 and P2 together state that the products are weak substitutes and have the connected substitute property, see \citet{BGH}. These properties are satisfied by most pricing models studied in the literature including linear demand models, MNL models, Exponomial choice models \citep{alptekinouglu2016exponomial}, envy-free pricing models \citep{rusmevichientong2006nonparametric} and any mixture of them. Moreover, P1 and P2 are also satisfied in pricing models for which even finding a price vector that guarantees some (positive) constant fraction of the optimal non-personalized profit $\cR^*$ is $\mathcal{NP}$-hard\footnote{Consider a demand model where each customer type has a preference list. Assume further that consumers remove product $i$ from the list if its price is higher than $r_i$, and then select their top choice from the remaining products, if any. This model satisfies P1 and P2 for $f=e$ and it is equivalent to an assortment optimization problem in which $r_i$ is the profit for product $i$. Thus, the strongest negative result to date about the inapproximability of assortment optimization ($\mathcal{NP}$-hardness to approximate to within a factor of $\Omega(1/n^{1-\epsilon})$ for every $\epsilon>0$ \citep{aouad2018approximability}) carries over to the pricing models studied in this paper.}.
If $d^j(p)$ is  differentiable in $p$ for all $j \in M$, then the aggregate demand $d(p) = \sum_{j \in M}\theta_j d^j(p)$ is also differentiable. Then P2 implies that $\nabla d(p) e \leq 0$ where $\nabla d(p)$ is the Jacobian matrix. By P1, $\partial d_i(p)/\partial p_k \leq 0$ for all $i \neq k$, so together P1 and P2 imply that $\nabla d(p)$ is a P-matrix. As a result, its inverse exists and is non-negative and the demand function is injective. For a given $f \neq e$ we can transform the demand function via $d(p) \leftarrow \mbox{diag(f)}d(\mbox{diag}(f)^{-1}p)$ and one can verify that the test of weak substitution and  connected substitute properties are are equivalent to P1 and P2 above. Details of the proof can be found in the Appendix.

We next show how to construct reasonable choices of $f$ when $p^*$ is hard to find and the firm knows $\bar{p}^j, j \in M$. An economically motivated  choice is $f = \sum_{j \in M}\alpha_j \bar{p}^j$ with weights $\alpha_j = \theta_j \cR^j/\bar{\cR}, j \in M$. For such $f$, $\cR^f$ is called the \emph{Economic single factor} profit. While the economic factor works well in practice, it does not minimize $\rho$ among all possible vectors $f$. To find a robust $f$,  let $p^H_i := \max_{j \in M}\bar{p}_{ij}$, $p^L_i := \min_{j \in M}\bar{p}_{ij}$. Define the robust  factor by $f^*_i : =  \sqrt{p^L_i p^H_i}, i \in N$, and the \emph{Robust single factor} profit by $\cR^{f^*}$. Let $\rho_i := p^H_i/p^L_i ~~ \forall ~i \in N$, and $\rho^* = \max_{i \in N} \rho_i$.

%we define $q^f_{\min}:= \min_{i \in N, j \in M}p_{ij}/f_i$, $q^f_{\max} := \max_{i \in N, j \in M}p_{ij}/f_i$,  and $\rho^f := q^f_{\max}/q^f_{\min}$. $p^H_i := \max_{j \in M}p_{ij}$, $p^L_i := \min_{j \in M}p_{ij}$.

\begin{thm}\label{thm:robust}
$\rho \geq \rho^*$. Moreover, $\rho^*$ is attained by $f^*$.
\end{thm}

\begin{proof}
Let $f$ be any positive vector. Then for any $k \in N$
$$\rho =  \frac{\max_{i \in N, j \in M} \bar{p}_{ij}/f_i}{\min_{i \in N, j \in M} \bar{p}_{ij}/f_i} \geq \frac{\max_{j \in M} p_{kj}/f_k}{\min_{i \in N, j \in M} \bar{p}_{ij}/f_i} \geq
\frac{\max_{j \in M} p^*_{kj}/f_k}{\min_{j \in M} p^*_{kj}/f_k} = \rho_k.$$
Since this holds for all $k \in N$, it follows that $\rho \geq \rho^*$.  We next show that $\rho^*$ is attained by $f^*$. By construction,  $p_{ij}/f^*_i \in  [(p^L_i/p^H_i)^{0.5}, (p^H_i/p^L_i)^{0.5}]$ with the two bounds attained. Consequently, $q_{\max} = \max_{i \in N} (p^H_i/p^L_i)^{0.5}$, and
and $q_{\min} =  \min_{i \in N} (p^L_i/p^H_i)^{0.5}$.
Clearly the product that attains the maximum in $q_{\max}$ also attains the minimum in $q_{\min}$, so
$\rho = \max_{i \in N} p^H_i/p^L_i = \max_{i \in N} \rho_i = \rho^*$.
\end{proof}

\section{Clustering Consumer Types}\label{sec:clustering}

Suppose that $m$ is large and the price vectors $\bar{p}^j, j \in M$ are dissimilar resulting in a large $\rho$ and therefore a poor performance guarantee.  The firm can potentially improve the worst case performance if it can partition $M$ into $K$ collectively exhaustive and mutually exclusive clusters, so  $M = \cup_{k = 1}^KM_k$. We have already dealt with the case $K =1$, while $K = m$ corresponds to personalized pricing. The problem is interesting for $1 < K < m$.

For a given partition, and a given positive vector $f^k$ for cluster $k$, the worst case performance is $1 + \ln(\rho(M_k))$ where $\rho(M_k) $ is the corresponding worst case ratio for cluster $M_k$ and factor $f^k$ and $\rho^*(M_k) \leq \rho(M_k)$ is the minimal ratio corresponding to the robust factor for cluster $M_k$.  More formally, $\rho^*(M_k) = \max_{i \in N} \rho_{ik}$, where
$\rho_{ik} = p^H_{ik}/p^L_{ik}$, $p^H_{i,k} := \max_{j \in M_k}\bar{p}_{ij}$, and $p^L_{i,k} := \min_{j \in M_k}\bar{p}_{ij}$. We define the clustering problem as finding a partition to $\min_{M_1,\ldots,M_K} \max_k \rho^*(M_k)$.

This problem is known in the literature as that of  minimizing the maximum inter-cluster distance in the context of graphs \citep{gonzalez1985clustering} and also as the bottleneck problem \citep{hoschbaumshmoys}. To see that equivalence, consider a graph $G$ with vertices $(i,j)$ for all $i \in N$ and $j \in M$. There are edges between any two nodes that share the same product,  say $(i,j)$ and $(i, l)$. The {\em distance} between the two adjacent nodes in the network is given by $\max(\bar{p}_{ij}, p^*_{il})/\min(\bar{p}_{ij}, p^*_{il})$. The reader can confirm that the maximum distance for a graph $G$, is equal to $\rho^*(G)$, and that the distance satisfies the triangle inequality.

Fortunately, there is a 2-factor approximation polynomial time algorithm for this problem, which is the best possible unless $\mathcal{P}=\mathcal{NP}$, see \citet{gonzalez1985clustering} and \citet{hoschbaumshmoys}. We also used the $k$-means clustering heuristic and compare their performance. We remark that once the clusters are formed, each cluster will price along a positive vector which may or may not be the robust choice for that cluster. This is because frequently the economic factor performs better than the robust factor even though the robust factor gives the best performance guarantee.

%We sketch the farthest segment first (FSF) $K$-clustering algorithm here:
%
%\begin{enumerate}
%\item Segment 1 goes to cluster 1.
%\item Choose the segment $i$ ($i \neq 1$) that would give you the worst $\rho$ if combined with the previous segment.
%\item  Put segment $i$ into another cluster (eg. 2).
%\item Choose another segment $j$ ( $j \neq 1$; $j \neq  i$) that would give you the worst $\rho$ if combined either with segment $1$ or segment $i$.
%\item Put segment $j$ into another cluster.
%\item Keep doing this until you have a single segment in each of the $K$ clusters.
%\item For each remaining segment, place it into the cluster that minimizes the marginal increase of the overall $\rho$.
%\end{enumerate}

\section{Applications}
\label{sec:app}

In this section we discuss several applications to our results, including linear demands, the latent class MNL, and Non-Linear Pricing.

\subsection{Linear Demands}\label{sec:linear_demand}

We first briefly review the representative consumer problem that results in the linear demand model. The task of the representative consumer is to solve the problem $\max_{q \geq 0} [(u-p)'q - q'Sq]$ where $u$ is the vector of gross utilities, $u-p$ is the vector of net utilities, and $S$ is a positive definitive matrix. The solution that ignores the non-negativity constraints yields $q^*  = B(u-p) = a - Bp$ where $a := Bu$ and $B := (S +S')^{-1}$.  Notice that by construction $B$ is symmetric and positive definitive. Let $P := \{p \geq 0: a - Bp \geq 0\}$. We assume that $a$ has positive components. Then $d(p) = a - Bp$ for all $p \in P$. Maximizing $R(p) = p'(a - Bp)$ yields $p^* = 0.5B^{-1}a = 0.5u$, $d(p^*) = 0.5a$ and $R(p^*) = 0.25 u'Bu > 0$, so $p^* \in P$. For $p \notin P$, the solution to the representative consumer's problem is equivalent to solving the linear complementarity problem $y \geq 0$, $d(p-y) \geq 0$,   and $y'd(p-y) = 0$, see \citet{gallegotopaloglu}.

Suppose that for all $j \in M$, $d^j(p) = B_j(u^j-p) = a^j -B_jp$ for $p \in P_j : = \{p \geq 0: a^j - B_jp \geq 0\}$.
where $u^j$ and $a_j$ are positive vectors, and $B_j$ is a symmetric positive definitive matrix. Then $\bar{p}^j = 0.5B^{-1}_ja^j = 0.5 u^j$ is in $P_j$,  and $\bar{\cR} = \sum_{j \in M}\theta_j \cR_j$ can be computed without problems.

%Let $d(p) = \sum_{j \in M}\theta_jd^j(p) = a - Bp$, then $p^*  = 0.5 B^{-1}a$
%
%$p^* = \arg\max p'(a-Bp)$ and $R(p^*) = p^*'(a- Bp^*)$. Also, for any $f$ let $q^* = \arg \max {q \geq 0} qf'(a - qBf)$ which is a simple quadratic problem.

\begin{cor} Suppose  that $B_j$ has non-positive off-diagonal elements for all $j \in M$. Then
Therorem~\ref{thm:sf} holds for positive vectors $f$ such that $B_jf \geq 0$ for all $j \in M$.  Moreover, if $B_j$ is an M-matrix for all $j \in M$ then the Theorem~\ref{thm:sf} holds for all positive $f$.
 \end{cor}

\begin{proof}
From the assumptions of the corollary,  P1 and P2 holds for all $j \in M$ . Since P1 and P2 are sufficient for A1, the result of Theorem~\ref{thm:sf} hold.  Moreover, if $B_j$ is an M-matrix then $B^{-1}_j$ has positive components. This implies that the vector $\epsilon B^{-1}_j e$ has positive components that can be made arbitrarily small, so for any vector positive vector $f$ we can find an $\epsilon > 0$ such that have $f \geq \epsilon B^{-1}_j e$. Multiplying both sides by $B_j$ shows that $B_j f \geq \epsilon e > 0$.
\end{proof}

At this point we know that $\bar{\cR} \leq \beta \cR^f \leq \beta \cR^*$ where $\cR^f$ and $\cR^*$ are interpreted as $\cR^f = \max_{q \geq 0}qf'(a- qBf)$ and $\cR^* = \max_{p \geq 0}q'(a- Bp)$ where $a  = \sum_{j \in M} \theta_j a^j$ and $B = \sum_{j \in M} \theta_j B_j$. The caveat is that optimal solutions to this problems, say $q^*f$ and $p^*$, may be outside $P_j$ for some $j \in M$ resulting in negative demands for some products for some customer types. This cast as question of whether the profits from pricing at $q^*f$ or at $p^*$ will actually satisfy the guarantees of Theorem~\ref{thm:sf}. We will show that $\bar{\cR} \leq \beta \cR^f$ and $\bar{\cR} \leq \beta \cR^*$ continue to hold even after the adjustments required to ensure that all demands are non-negative. To see this in a generic form, we will argue that the actual profit when $d(p)$ has negative components is at least as large as $R(p) = p'd(p)$.
\begin{prop}
The profit under the representative consumer model is equal to $R(p) + p'By \geq R(p)$ when $p \notin P$.
\end{prop}

\begin{proof}
 The expected profit associated with a vector $p \notin P$ is given by $p'd(p-y) = p'd(p) + p'By = R(p) + p'By$  where $y \geq 0, d(p-y) \geq 0$ and $y'd(p-y) = 0$. The complementary slackness condition $y'd(p-y) = 0$ can be written as $y'[a - Bp + By] = 0$. By the symmetry of $B$, $p'By = y'Bp = y'[a + By] = y'a + y'By  \geq 0$. The inequality follows from $a \geq 0, y \geq 0$ and $y'By \geq 0$ since $B$ is positive definitive.
\end{proof}

We can now apply the result to each customer type that has negative demands at $p$. In particular, if there is a customer type with negative demands at $q^*f$, then the firm will see profits $R_j(q^*f) + q^*f'B_jy^j \geq R_j(q^*f)$ from type $j$ customers, so $\cR^f  = \max_qR(qf)$ is a lower bound on the aggregate profit over all customer types at $q^*f$. Consequently, the profit from the single factor model is at least $\cR^f \geq \bar{\cR}/\beta$.  In a similar way, $\cR^* = \max_p R(p)$ is a lower bound of the aggregate profits at $p^* \in \arg\max R(p)$, so the profit under $p^*$ is at least $\cR^* \geq \bar{\cR}/\beta$. One must be aware, however, that for multiple customer types, $p^* = 0.5 B^{-1}a$ is not necessarily optimal if there are customer types with negative demands at this price vector. To find a true optimal solution to the problem the firm needs to solve
$\max_{p\geq 0, y^j, j \in M}\sum_{j \in M} \theta_j[R_j(p) + p'B_jy^j]$ subject to $d^j(p-y^j) \geq 0, y^j \geq 0$ and $y^j_id_{ij}(p-y^j) = 0$ for all $i \in N, j \in M$. The solution $p^* = 0.5B^{-1}a$ together with the corresponding $y^j$s that solve the linear complementarity problem for market segments with negative demands is only a heuristic for this problem, so our results continue to hold if the more complex problem is solved, with a similar more sophisticated program for pricing along a factor $f$.

%Moreover, if $B^{-1} \geq 0$, then Proposition 1 holds for all positive vectors $f$.   We remark that a sufficient condition for $B^{-1} \geq 0$ is that $B$ is an $M$-matrix, or equivalently that  $B_{ii} + \sum_{k \neq i}B_{ik} \geq 0$ for all $i \in N$ with $B_{ik} \leq 0$ for all $i \neq k$.
Figure \ref{fig_linear_demand} reports a series computational results in order to evaluate the performance of different pricing strategies. For each value of $n$ and $m$, we generated 20 random instances and we reported the average profit as a percentage of the maximum profit that can be obtained using personalized pricing. The weight of each segment $j$ was set to $\theta_j=x_{j}/\sum_{l=1}^mx_l$ where each $x_{l}$ is a uniform random number between 0 and 1. For each segment $j$ we randomly generated the matrix $B_j$ ensuring it is symmetric and positive definitive and satisfies P1 and P2. The vector $a^j$ was generated uniformly random from $(0,1]^n$. The percentages shown are the average percentage of the profit obtained with respect to the best personalized pricing strategy. As we can see, the economic factor outperformed the robust factor, and both performed significantly better than uniform pricing based on $f = e$. The lower right-hand table for the optimal price uses the heuristic $p^* = 0.5B^{-1}a$ adjusting demands by solving the complementary slackness for customer types with negative demands. It performance is similar to that of the economic factor.

\begin{figure}[H]
\centering
\includegraphics[scale=0.30]{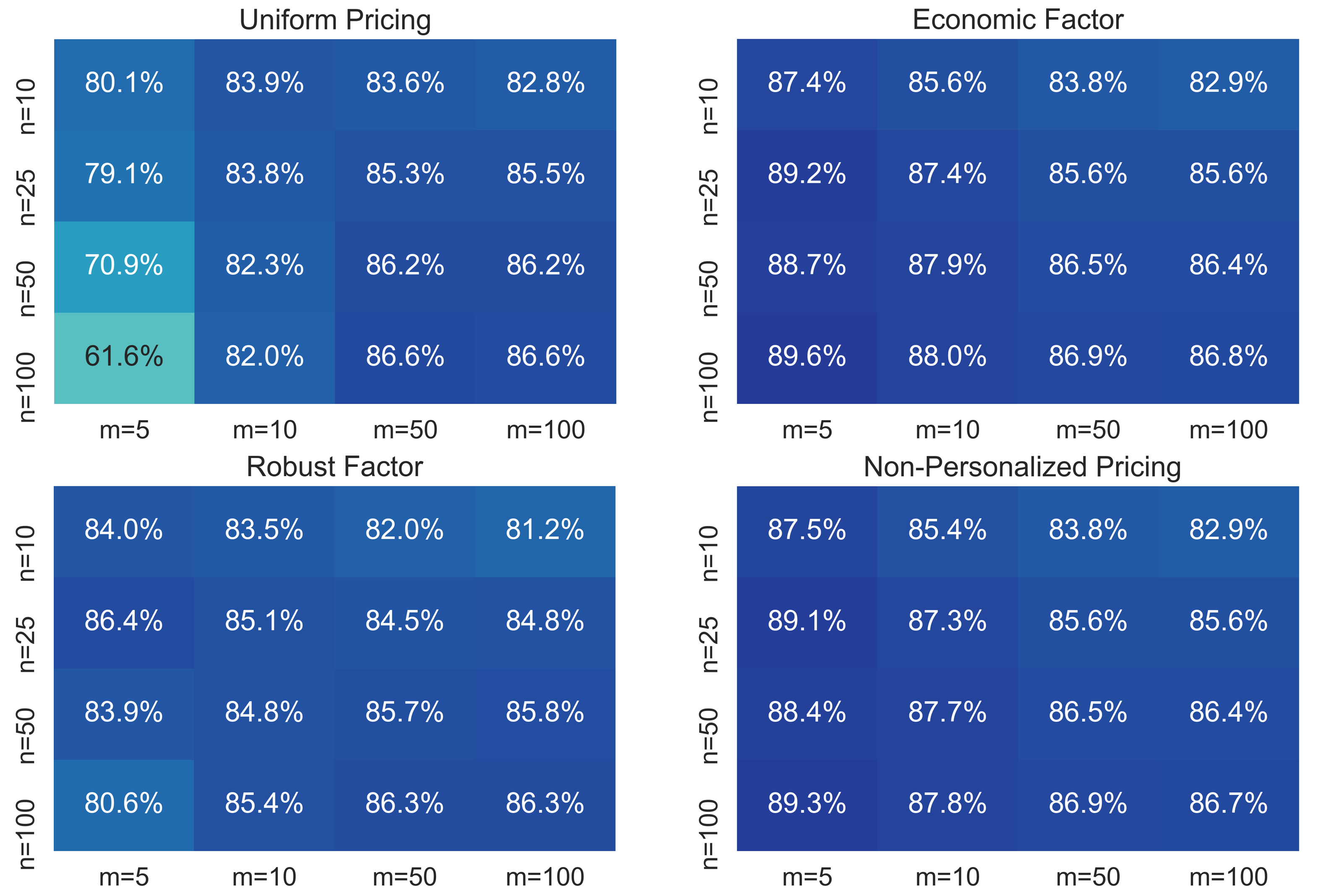}
\caption{Average performance of different pricing strategies under the linear demand model.}
\label{fig_linear_demand}
\end{figure}

Figure \ref{fig_linear_demand_clustering} reports another set of experiments to quantify the advantages of clustering consumer segments into two clusters ($k=2$) under the economic and the robust pricing strategies. Two clustering algorithms were implemented. The first one is the standard k-means algorithm in which each segment $j$ was assigned the price vector $\bar{p}^j$ as its representative point in an $n$-dimensional space. The second is the \emph{farthest point first} (FPF) proposed by \cite{gonzalez1985clustering} where the distance matrix is set as described in Section \ref{sec:clustering}. For these experiments the number of consumer segments was set to $m=6$. As can be seen, the best combination was the economic factor coupled with $k$-means and the worse was the robust factor with FPF.

\begin{figure}[H]
\centering
\includegraphics[scale=0.30]{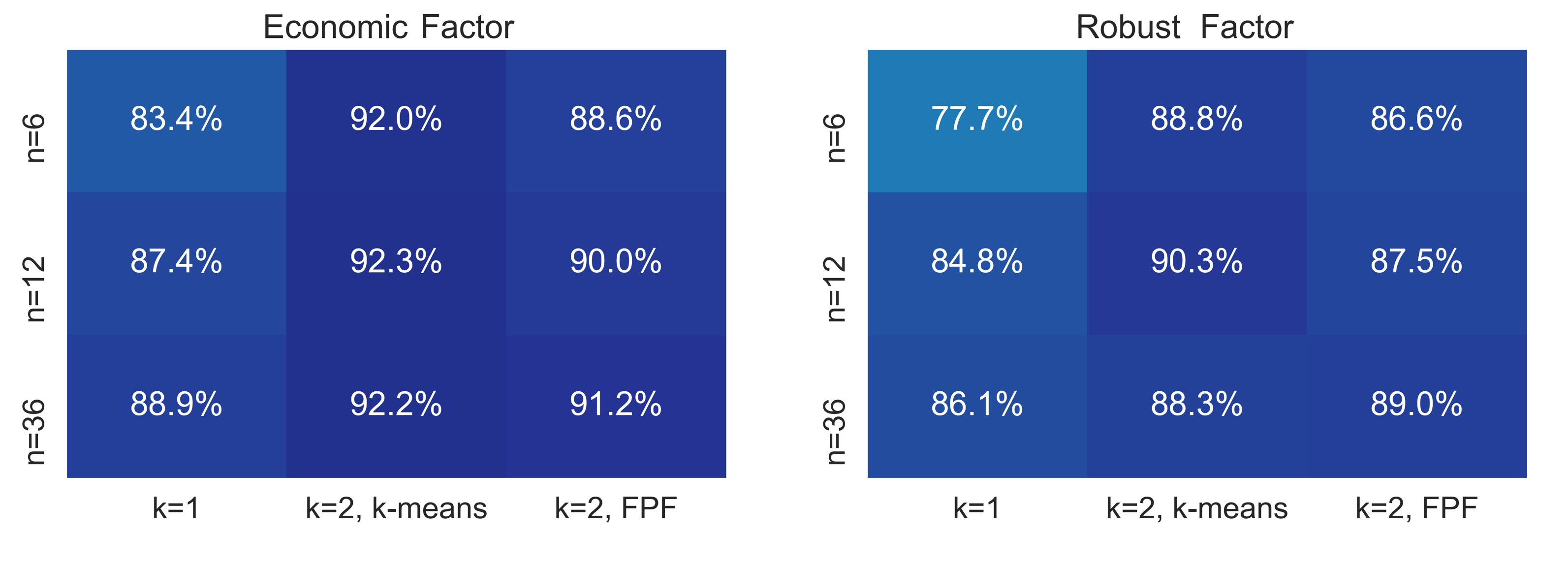}
\caption{Average performance of Economic and Robust pricing strategies for instances with 6 segments with and without clustering with $k=2$.}
\label{fig_linear_demand_clustering}
\end{figure}

\subsection{Latent Class MNL}

Suppose that $d^j(p)$ is the expected demand from an MNL model, so
$$d_{ij}(p) = \frac{\exp(a_{ij} - b_{ij}p_i)}{1 + \sum_{k \in N} \exp(a_{kj}- b_{kj} p_k)}~~~~\forall~~~k \in N.$$
The matrix of partial derivatives is given by
$\nabla d^j(p) = \mbox{diag}(b^j) \left[d^j(p)d^{j}(p)' - \mbox{diag}(d^j(p))\right]$.
Since the off-diagonal elements are non-negative we see that $d_{ij}(p)$ is increasing in $p_k, k \neq i$ and P1 holds.  P2 hold for all positive vectors $f$ such that $\nabla d^j(p)f \leq 0$, or equivalently for all positive vectors $f$ such that $d^j(p)'f \leq \min_{i \in N}f_i$ for all $j \in M$. We can scale $f$ without loss of generality so that $\min_{i \in N}f_i = 1$, which reduces the condition to $\sum_{i \in N} f_id_{ij}(p) \leq 1$. This clearly holds for $f = e$ on account of $\sum_{i \in N}d_{ij}(p) = 1 - d_{0j}(p) \leq 1$. For any positive $f$ the condition reduces to $\sum_{i \in N}(f_i -1) \exp(a_{ij} - b_{ij} p_i) \leq 1$ for all $j \in M$. We remark that A1 requires the condition to hold only at $p = \bar{p}^j$, for $j \in M$. We know that for each consumer type the optimal price is of the form $\bar{p}_{ij} = 1/b_{ij} + m_j$ where $m_j$ is the adjusted markup for product type $j$ consumers, so the condition is easy to check.  In particular, if $b_{ij} = b_j$ for each $i \in N$, and $j \in M$, then both the economic and the robust factors can be taken to be equal to $e$.

\begin{cor}
For the LC-MNL $\bar{\cR} \leq \beta\cR^e \leq \beta \cR^*$
without any further conditions since P1 and P2 hold for $f = e$  for all MNL models with arbitrary price sensitivities. In addition, if the $b_{kj} = b_j$ is independent of $k \in N$ for each $j \in M$, then the economic and the robust factors are equivalent to $e$.  Finally,  $\bar{\cR} \leq \beta \cR^f \leq \beta \cR^*$
holds for all positive vectors $f$ such that $d^j(\bar{p}^j)f \leq \min_{i \in N}f_i$ for all $j \in M$.
\end{cor}

We tested the performance of the different heuristics under the LC-MNL model and reported the results in Figure \ref{fig_lc-mnl_demand}. The percentages shown are the average percentage of the profit obtained with respect to the best personalized pricing strategy. For each value $n$ and value $m$ reported, we generated 20 random instances with $n$ products and $m$ segments. The mean utility of product $i$ to consumer segment $j$ is $u_{ij} = a_{ij} - b_{ij}p_i$ where the intrinsic product utility $a_{ij}$ were randomly chosen following a procedure proposed by \cite{rusmevichientong2014assortment} \footnote{Specifically, the intrinsic utility of product $i$ for consumer segment $j$ is defined as $a_{ij}:=\ln((1-\sigma_i)v_{ij}/n)$ with probability $p=0.5$ and $a_{ij}::= \ln((1+\sigma_i)v_{ij}/n)$ otherwise. The values $v_{ij}$ and $\sigma_i$ are realizations from a uniform distribution $[0,10]$ and $[0,1]$ respectively.} and the (segment and product dependent) price sensitivities $b_{ij}$ were randomly chosen from a symmetric triangular distribution between 0 and 2. We can observe that while uniform pricing does relatively well (obtaining between 60.6\% to 91\% of the optimal personalized profit) it is surpassed by the Economic and Robust strategies with get at least 76.3\% and 75\% respectively. The values for the non-personalized pricing are not necessarily the optimal ones \footnote{This is an NP-hard problem.} since they were obtained using a multi-variable non-linear solver in Python. This strategy requires much broader computational resources than the other three methods which simply rely on a single variable optimization. For example, when $n=100$ the solver took on average over 22 times more time than any of the other 3 strategies.

\begin{figure}[H]
\centering
\includegraphics[scale=0.30]{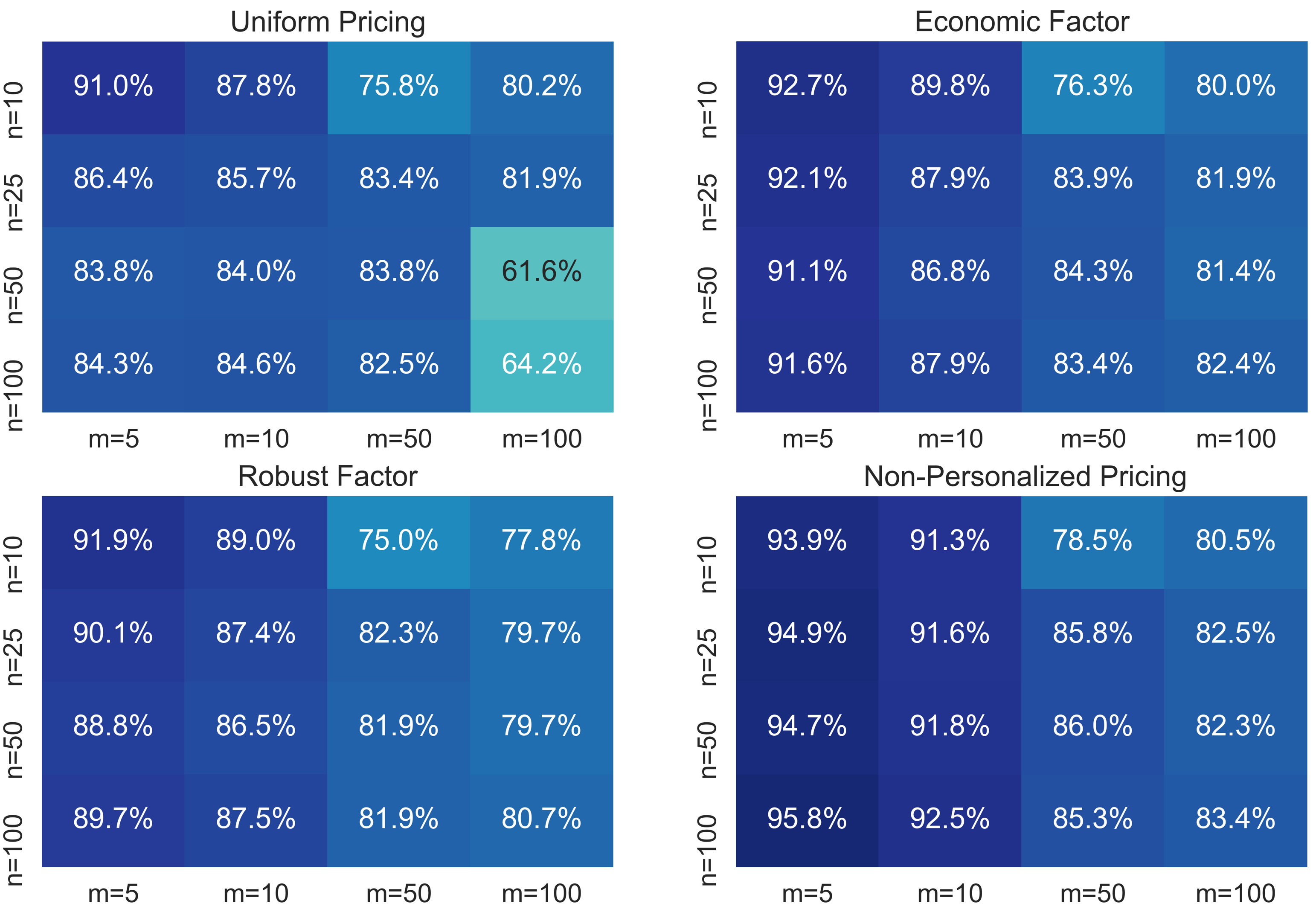}
\caption{Average performance of different pricing strategies when the demand model is the LC-MNL.}
\label{fig_lc-mnl_demand}
\end{figure}

Similarly to the clustering results for the linear demand model, Figure \ref{fig_lc-mnl-clustering} reports computational results that quantify and compare the benefits of clustering consumer segments using k-means and FPF under the Economic single factor and the Robust single factor pricing strategies. The values represent the average percentage of the profit obtained with respect to the best personalized pricing strategy. As in the linear demand model, $m=6$ for all these experiments. As can be seen, the best combination was the economic factor coupled with $k$-means and the worse was the robust factor with FPF.

\begin{figure}[H]
\centering
\includegraphics[scale=0.30]{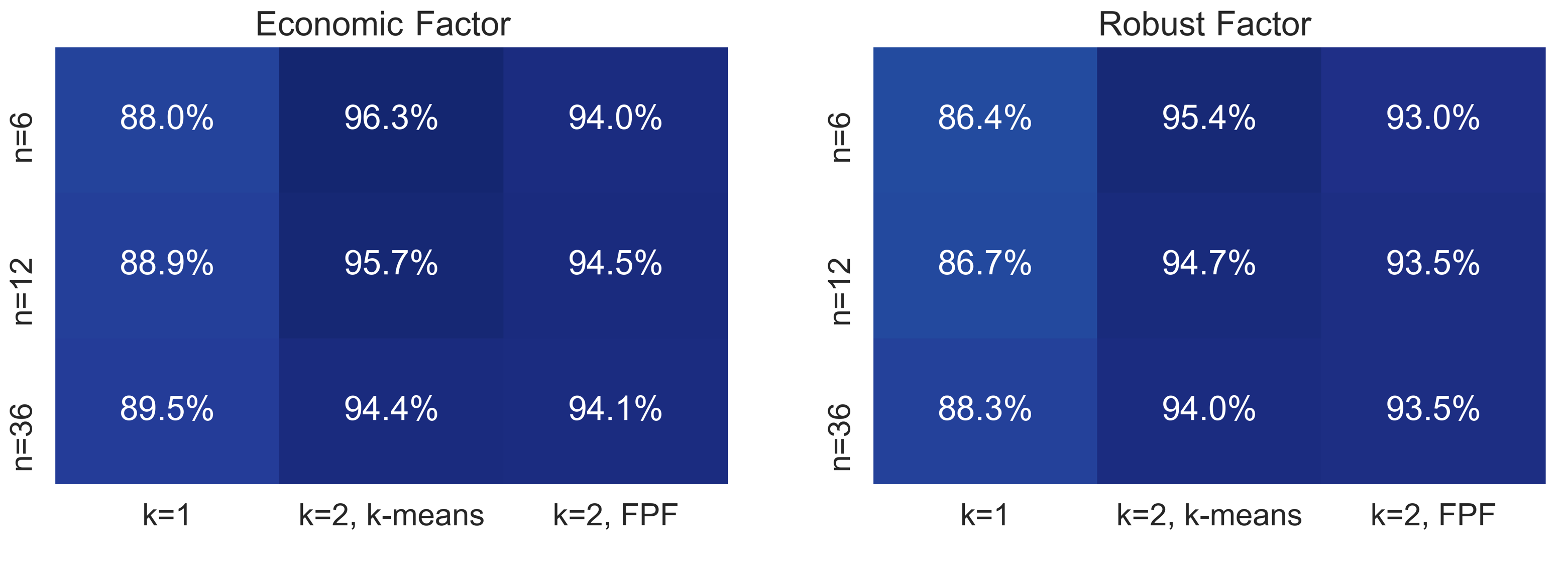}
\caption{Average performance of Economic and Robust pricing strategies for instances with 6 segments with and without clustering with $k=2$.}
\label{fig_lc-mnl-clustering}
\end{figure}

\subsection{Linear Pricing Versus Non-Linear Personalized Pricing}
We now consider the non-linear pricing scheme where the firm sells a single product in different bundle sizes- for a broad overview of non-linear pricing see \cite{wilson1993nonlinear} and \cite{oren2012nonlinear}. Let $p_i$ be the price of a size $i \in N$ bundle and $d_i(p)$ is the demand for a size $i$ bundle at the price vector $p$. Let $\cR^* = \max_pR(p)$ yielding an optimal non-linear price schedule. Let $f$ be a vector with components $f_i = i, i \in N$. Then $\cR^f := \max_q qf'd(qf)$ corresponds to the linear price schedule $p_i = iq, i \in N$. Theorem 2 holds if $d_i(p)$ is increasing in $p_k$, $k \neq i$ and if $f'd(p)$ is decreasing in $p$.  We remark that $f'd(p) = \sum_{i \in N} id_i(p)$ is the total number of units demanded at price $p$, so the condition is that the total number of units demanded goes down if the price of any bundle is increased.

As an example, suppose that $d(p) = a -Bp$ and $B$ is an $M$-matrix, then $Bf \geq 0$ for all $f$ and in particular for $f_i = i$. Let $v := B^{-1}a$. Then $p^* = v/2$. If $v_i$ is increasing concave then $q_{\min} = v_n/n$ and $q_{\max} = v_1$ resulting in  $\beta = 1 + \ln(nv_1/v_n)$.

To our knowledge this is the first result that gives a performance guarantee for linear versus non-linear pricing, but we can go further as Theorem~\ref{thm:sf} works for the personalized version as well. More precisely, if $d^j(p)$ is the demand vector for bundles of size $i \in N$ for every $j \in M$, $d_{ij}(p)$ is increasing in $p_k$ for all $k \neq i$, and $f'd^j(p)$ is decreasing in $p$ for all $j \in M$ then Theorem~\ref{thm:sf} applies and bounds how much better personalized non-linear pricing can be relative to linear pricing. Theorem~\ref{thm:sf} can also be used to bound the performance of non-personalized non-linear pricing schemes relative to personalized non-linear pricing schemes. We summarize the results for personalized non-linear pricing here.

\begin{cor}
If $d_{ij}(p)$ represent the demand for bundles of size $i \in N$ in market segment $j \in M$ where $p_i$ is the price of a size $i \in N$ bundle, $d_{ij}(p)$ is decreasing in $p_k, k \neq i$ and $\sum_{i \in N}id_{ij}(p)$ is decreasing in $p$ for all $j \in M$, then
$$\bar{\cR} \leq \beta \cR^f \leq \beta \cR^*$$
where $\cR^*$ is the profit from the optimal non-personalized non-linear pricing policy and $\cR^f$ is the optimal under non-personalized linear pricing policy.
\end{cor}

We remark that $f$ here was selected as $f_i = i, i \in N$ for the purpose of comparing  a common linear price schedule for all customer types to optimal personalized non-linear pricing. One can instead use the robust $f$ or the economic $f$ to obtain a potentially better common (non-linear) price schedule.

Figure \ref{fig_non_linear_demand} reports a computational results about the performance of different pricing strategies for a non-linear pricing problem as described above. Each consumer segment follows a linear model as explained in Section \ref{sec:linear_demand}. The matrix $B_j$ associated to segment $j$ was generated in the same way as for the experiments of Section \ref{sec:linear_demand} whereas instead of generating a random vector $a^j$, we produced a random vector of utilities $u$ satisfying that $u_{i+1}>u_i$ and $u_{i+1}/(i+1) < u_{i}/i$ for all $i \in [n-1]$. We generated 20 instances for each value of $m$ and each maximum bundle size ($n$). As can be seen from the tables, linear pricing performance relatively well for $n  = 10$ but deteriorates as the maximum bundle size increases achieving only about 77\% of optimal personalized non-linear pricing. The results do not seem to be sensitive to the number of customer types. The robust factor slightly outperforms the economic factor, and the non-personalized non-linear pricing strategy is very close to the optimal personalized strategy with little need for clustering types.

\begin{figure}[H]
\centering
\includegraphics[scale=0.30]{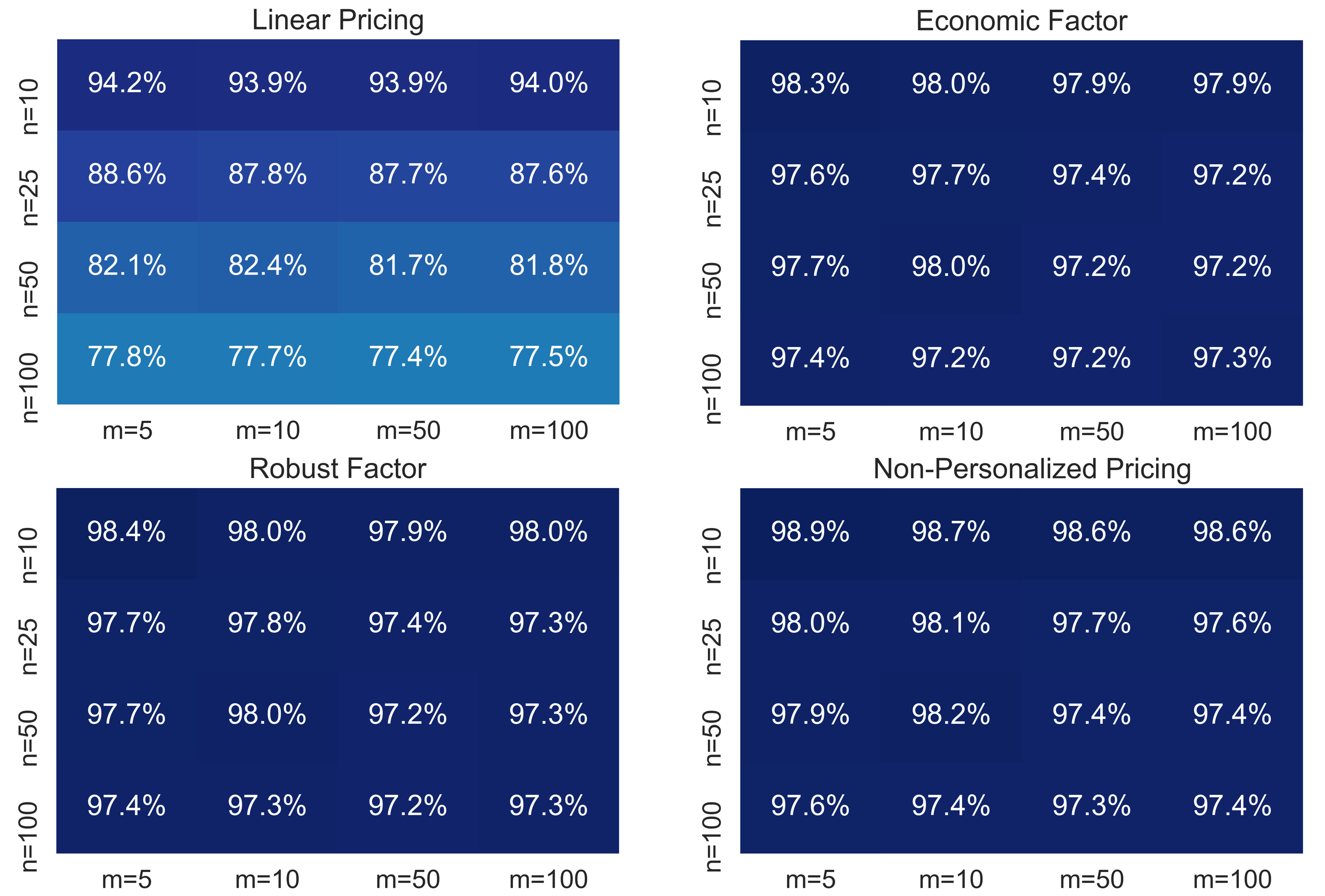}
\caption{Performance of different pricing methods under non-linear pricing. Here $n$ represents the maximum bundle size.}
\label{fig_non_linear_demand}
\end{figure}

\subsection{Bundle-size pricing as an approximation to mixed bundling}

Consider a set of $n$ products that can be sold as bundles. Suppose that given prices for each of the $2^{n-1}$ non-trivial bundles, the firm can obtain the demand for each of the bundles. Selecting the bundle prices to maximize profits is know as the mixed bundle problem. Suppose  that $p(x), x \in \{0,1\}^n$ is an optimal solution to the mixed bundle problem. We may wonder about the performance of several pricing strategies relative to mixed bundling. The simplest strategy is to use uniform pricing. This gives rise to $f(x) = e'x$. A second strategy, known as component pricing is to set $f(x) = p'x$ where $p_i$ is the price of component $i$. Finally, we can have a non-linear function $f(x) = g(e'x)$ where $g$ is an increasing non-linear function. Notice this form yields the same  price factor for all bundles of size $e'x$.  This is known as bundle-size pricing and contains uniform pricing as a special case if $g(e'x) = e'x$. In all cases, the firm will find and optimal $q$ for the pricing strategy $q \cdot f(x)$ for size $x$ bundles, \citet{ChenLeslieSoren2011bunde} shows through extensive numerical studies that bundle-size pricing can do a good job of approximating the benefits of the more complicated mixed bundling strategy but to our knowledge there are no theoretical work on tight bounds. Let $q_{\max} = \max_{x \neq 0} p(x)/f(x)$ and $q_{\min} = \min_{x \neq 0} p(x)/f(x)$. Under mild conditions on the demand function for bundles (A0 and A1, or A0 together with P1 and P2) we obtain performance guarantees $\cR^* \leq \beta \cR^f$ where here $\cR^*$ is the optimal profit under mixed bundling and $\cR^f$ is the optimal pricing along the vector $f(x), x \in \{0,1\}^n$, where $\beta = 1 + \ln(\rho)$ where $\rho = q_{\max}/q_{\min}$. As an example, if $n = 2$ then the non-trivial bundles are $e_1, e_2$ and $e = e_1 +e_2$ where $e_i$  is the $i$th unit vector in $\Re^2$. If $p(e_1) = 1, p(e_2) = 2$ and $p(e) = 2.5$ and we use $f(1) = 1$ and $f(2) = 2$ for bundles of size 1 and 2, then $q_{\min} = 1$ and $q_{\max} = 2$ so $\rho = 2$ and $\beta = 1 + \ln(2) \simeq  1.693$.  The theory also holds if there are multiple types and the firm uses personalized mixed pricing for each type. The only difference is that the definition of $q_{\min}$ and $q_{\max}$ has to be over all bundles and all types. It is also possible to fit a robust or an economic mixed bundle strategy and as long as the conditions A0 and A1 hold the bound from Theorem~1 holds.  If the model is linear or latent class MNL we will get similar numerical results with the exception that the bundles are interpreted as products and the vector $f$ is either the bundle size or either the economic or robust factor in the case of multiple market segments.

\section{Conclusions and Future Research} This paper presents tight performance guarantees for multi-product single factor pricing relative to personalized pricing with applications to a variety of demand models. The results apply to $d_i(p_i, p_{-i})$, where $d_i$ is the demand vector  for firm $i$  at price vector $p_i$ given that competitors offer price $p_{-i}$ for their own goods provided A1 or P1 and P2 hold for fixed $p_{-i}$. This opens the door to study competition  under a variety of pricing scenarios.

\section{Appendix}

\begin{proof}  For convenience, we first consider the single customer type case dropping the index $j$. Let $p^*\in \arg\max p'd(p)$ and $\delta_i(q) = 1$ for $q \leq p^*_i/f_i$ and $0$ otherwise.
 Since $qf \leq \max(p^*,qf)$, P2 implies that
 $$f'd(\max(p^*,qf)) \leq f'd(qf).$$

The move from $\max(p^*,qf)$ to $qf$ decreases the price of products for which $\delta_i(q) = 1$.  By P1 this has a negative effect on the demand of products for which $\delta_i(q) = 0$. Thus,
$$\sum_{i \in N}f_id_i(qf)(1-\delta_i(q))  \leq \sum_{i \in N} f_id_i(\max(p^*,qf))(1-\delta_i(q)) .$$
Consequently,
$$\sum_{i \in N}f_id_i(qf)\delta_i(q) \geq \sum_{i \in N}f_id_i(\max(p^*,qf))\delta_i(q).$$
Moreover, moving from $\max(p^*,qf)$ to $p^*$ decreases the prices of products for which $\delta_i(q) = 0$, so by P1
$$d_i(\max(p^*,qf))\delta_i(q) \geq d_i(p^*)\delta_i(q)~~~~\forall~~~~i \in N.$$
Multiplying by $f_i$, adding and collecting the inequalities we obtain
$$f'd(qf)  \geq \sum_{i \in N}f_id_i(\max(p^*,qf))\delta_i(q) \geq \sum_{i \in N}f_id_i(p^*)\delta_i(q).$$
This completes the proof for a single customers type, and implies  under the stated assumptions that
$$\sum_{i \in N}f_id_{ij}(qf) \geq \sum_{i \in N}f_id_{ij}(\bar{p}^j)\delta_{ij}(q)~~~\forall~~~j \in M.$$
Multiplying by the weights $\theta_j$ and adding over $M$ yields $H(q) \geq G(q)$ which is A$1$.
\end{proof}

\section*{Acknowledgements}
We would like to thank Adam Elmachtoub, Pin Gao, Wentau Lu, Preston McAfee, Shmuel Oren, and Zhuodong Tang for their help and suggestions, as well as Daniel Aloise for pointing us to the minimax diameter clustering problem and the FPF heuristic used in this paper.

\bibliographystyle{plainnat}
\bibliography{references}

\end{document}